\documentclass{article}
\usepackage{fullpage}
\usepackage[utf8]{inputenc}
\usepackage{amssymb}
\usepackage{amsmath}
\usepackage{graphicx}
\usepackage{url}
\usepackage{multirow}
\usepackage{algorithmic}
\usepackage{tikz}

\newcommand{\seq}[1]{\left\langle #1\right\rangle}

\newcommand{\floor}[1]{\left\lfloor #1\right\rfloor}
\newcommand{\abs}[1]{\left| #1\right|}

\newcommand{\set}[1]{\left\{ #1\right\}}

\newcommand{\Def}{:=}

\newcommand{\wrt}{w.r.t. }

\newcommand{\realrange}[2]{\left[#1, #2\right]}

\newcommand{\unitrange}[2]{\realrange{0}{1}}

\newcommand{\Oh}[1]{\mathrm{O}\!\left( #1\right)}

\newcommand{\Th}[1]{\Theta\!\left( #1\right)}

\newcommand{\llabel}[1]{\label{\labelprefix:#1}}
\newcommand{\labelprefix}{} 

\newcommand{\discussionsize}{\small}

\newcommand{\notiz}[1]{}
\newcommand{\frage}[1]{}

\newenvironment{code}{\noindent
\begin{tabbing}%
\hspace{2em}\=\hspace{2em}\=\hspace{2em}\=\hspace{2em}\=\hspace{2em}\=%
\hspace{2em}\=\hspace{2em}\=\hspace{2em}\=\hspace{2em}\=\hspace{2em}\=%
\kill}{\end{tabbing}}

\newcommand{\labelcommand}{}
\newcommand{\captiontext}{}
\newsavebox{\codeparam}
\newcounter{lineNumber}
\newenvironment{disscodepos}[3]{%
\renewcommand{\labelcommand}{#2}%
\renewcommand{\captiontext}{#3}%
\sbox{\codeparam}{\parbox{\textwidth}{#3}}%
\begin{figure}[#1]\begin{center}\begin{code}\setcounter{lineNumber}{1}}{%
\end{code}\end{center}\caption{\llabel{\labelcommand}\captiontext}\end{figure}}

{\end{disscodepos}}

\newcommand{\Function} {{\bf Function\ }}

\newcommand{\While}    {{\bf while\ }}

\newcommand{\Do}       {{\bf do\ }}

\newcommand{\ForEach}      {{\bf foreach\ }}

\newcommand{\If}       {{\bf if\ }}

\newcommand{\Then}     {{\bf then\ }}
\newcommand{\Else}     {{\bf else\ }}

\newcommand{\Return}   {{\bf return\ }}

\newdimen\endofsize\endofsize=0.5em
\def\endofbeweis{~\quad\hglue\hsize minus\hsize
                 \hbox{\vrule height \endofsize width
\endofsize}\par}

\newcommand{\ignore}[1]{}

\renewcommand{\frage}[1]{[{\sf#1}]\marginpar[\hfill$\Longrightarrow$]{
$\Longleftarrow$}}

\newcommand{\fn}[1]{\mathit{fn}(#1)}

\newtheorem{example}{Example}
\newtheorem{lemma}{Lemma}
\newtheorem{theorem}{Theorem}
\newcommand{\qed}{\hfill$\square$}
\newenvironment{proof}{\textit{Proof.}}{\qed}
\title{Contraction of Timetable Networks with Realistic Transfers}
\author{Robert Geisberger\\
Universität Karlsruhe (TH), 76128 Karlsruhe, Germany, {\tt geisberger@ira.uka.de}
}

\begin{document}

\maketitle

\begin{abstract}
We successfully contract timetable networks with realistic transfer times.
Contraction gradually removes nodes from the graph and adds shortcuts to preserve shortest paths.
This reduces query times to 1\,ms with preprocessing times around 6 minutes on all tested instances.
We achieve this by an improved contraction algorithm and by using a station graph model.
Every node in our graph has a one-to-one correspondence to a station and every edge has an assigned collection of connections.
Our graph model does not need parallel edges.
The query algorithm does not compute a single earliest arrival time at a station but a set of arriving connections that allow best transfer opportunities.

\end{abstract}

\section{Introduction}
Compared to road networks, query algorithms for timetable information in public transportation systems are still slow.
As multi-modal routing becomes more and more important, it is necessary to develop speedup techniques that work well.
On the one hand, previous research focused on modelling more and more features instead of fast algorithms.
On the other hand, current speedup techniques for time-dependent routing in road networks are very fast, e.g. \cite{d-tdsr-09,bdsv-tdch-09}, but they have some problems to process transportation networks.
So more research on fast timetable routing is necessary.
We successfully adapt a fast speedup technique for road networks.
The positive outcome is mainly due to the station graph model that looks much more natural and lacks a lot of problems that other models have.

\subsection*{Related Work}\label{s:related}

Public transportation networks have always been time-dependent, i.e. travel times depend on the availability of trains, buses or other vehicles.
So they are naturally harder than road networks, where simple models can be independent of the travel time and still achieve good results.
There are two intensively investigated approaches for modeling timetable information: the \emph{time-expanded} \cite{mw-pspof-01,mn-eatm-98,sww-daola-00,swz-umlgt-02}, and the so called \emph{time-dependent} approach \cite{bj-tnmaf-04,n-t-95,or-spmda-90,or-mwptd-91}.
Note that the time-dependent approach is a special approach to model the time-dependent information and is no umbrella term for all these approaches.
Both approaches answer queries by applying some shortest-path algorithm to a suitably constructed graph.
In the time-expanded approach, each node corresponds to a specific time event (departure or arrival), and each edge has a constant travel time.
In the time-dependent approach each node corresponds to a station, and the costs on an edge are assigned depending on the time in which the particular edge will be used by the shortest-path algorithm.

To model more realistic transfers in the time-dependent approach, \cite{bj-tnmaf-04} propose to model each platform as a separate station and add walking links between them.
In \cite{pswz-emtip-07}, a similar extension for constant and variable transfer times is proposed and described in more detail.
Basically, a station is expanded to a \emph{train-route graph} where no one-to-one correspondence between nodes and stations exists anymore.
A \emph{train route} is the maximal subset of trains that follow the exact same route, at possibly different times and do not overtake each other.
Each train route has its own node at each station and they are interconnected within a station with the given transfer times.
This results in a significant blowup in the number of nodes and creates a lot of redundancy information that is collected during a query.
Recently, another research group \cite{bm-somcs-09,bdgm-atdmc-09} independently proposed a model that is similar to ours.
They call it the \emph{station graph} model and mainly use it to compute all Pareto-optimal paths in a fully realistic scenario.
For unification, we will give our model the same name although there are some important differences in the details.
The most significant differences are that (1) they require parallel edges, one for each train route and (2) their query algorithm computes a label for each edge instead of each node.
Their improvement over the time-dependent model was mainly that they compare all labels at a station and remove dominated ones.

Speedup techniques are very successful when it comes to routing in time-dependent road networks, see \cite{dw-tdrp-09} for an overview.
However, there is only little work on speedup techniques for timetable networks.
Time-dependent SHARC \cite{dw-tdrp-09} uses the same scenario as we do and achieves query times of 2.4\,ms but with preprocessing times of more than 6 hours (we scaled by a factor of 0.5 compared to \cite{dw-tdrp-09} based on plain Dijkstra performance of our and their hardware).
Based on the station graph model, \cite{bdgm-atdmc-09} also applied some speedup techniques, namely arc flags that are valid on time periods and route contraction.
They could not use node contraction because they were too many parallel edges between stations.
Their preprocessing time (not scaled due to lack of comparable figures) is over 33 CPU hours resulting in a full day profile query time of more than 1 second (speedup factor 5.2).

We will show that a modified version of contraction hierarchies (CH) \cite{bdsv-tdch-09,gssd-chfsh-08}, a very successful speedup technique for road networks, will work in our scenario with realistic transfers.
It is solely based on node contraction: removing ``unimportant'' nodes and adding shortcuts to preserve shortest path distances.
A bidirectional query can then find shortest paths looking only at a few hundred nodes.

\section{Formal Description}\label{s:contraction}

We propose a model that is similar to the realistic time-dependent model introduced by \cite{pswz-emtip-07} but we keep a one-to-one mapping between nodes in the graph and real stations.

A \emph{timetable} consists of data concerning: \emph{stations} (or bus stops, ports, etc), \emph{trains} (or buses, ferries, etc), connecting stations, \emph{departure} and \emph{arrival times} of trains at stations, and \emph{traffic days}.
More formally, we are given a set of stations $\mathcal{B}$, a set of \emph{stop events} $\mathcal{Z_S}$ per station $S\in\mathcal{B}$, and a set of \emph{elementary connections} $\mathcal{C}$, whose elements $c$ are 6-tuples of the form $c=(Z_1, Z_2, S_1, S_2, t_d, t_a)$. Such a tuple (elementary connection) is interpreted as train that leaves station $S_1$ at time $t_d$ after stop $Z_1$ and the \emph{immediately next} stop is $Z_2$ at station $S_2$ at time $t_a$.
If $x$ denotes a tuple's field, then the notation of $x(c)$ specifies the value of x in the elementary connection $c$.
Two consecutive elementary connections, $c_1$ followed by $c_2$, where no transfer is required, share a stop event, i.e. $Z_2(c_1) = Z_1(c_2)$.
A stop event can not only be a consecutive arrival and departure of a train, but also the begin (no arrival) or the end (no departure) of a train.

The \emph{departure} and \emph{arrival times} $t_d(c)$ and $t_a(c)$ of an elementary connection $c \in \mathcal{C}$ within a day are integers in the interval $[0, 1439]$ representing time in minutes after midnight.
Given two time values $t$ and $t'$, $t \le t'$, the \emph{cycle difference}($t,t'$) is the smallest nonnegative integer $\ell$ such that $\ell \equiv t' - t$ (mod 1440).
The \emph{length}\ of an elementary connection $c$, denoted by \emph{length}($c$), is \emph{cycle difference}($t_d(c), t_a(c)$).
A timetable is valid for a number of $N$ \emph{traffic days}, and every train is assigned a bit-field of $N$ bits determining of which traffic days the train operates (for overnight trains the departure of the first elementary connection counts).
We will generally assume that trains operate daily.
At a station $S\in \mathcal{B}$, it is possible to \emph{transfer} from one train to another.
Such a transfer is only possible if the time between the arrival and the departure at the station $S$ is larger than or equal to a given, station-specific, \emph{minimum transfer time}, denoted by \emph{transfer}($S$).

\begin{example}
\label{example:elementary_connections}
\{(1, 1, $A$, $B$, 23:05, 0:55), (1, 1, $B$, $C$, 1:02, 2:57), (1, 1, $C$, $D$, 3:00,4:20)\} describe elementary connections of a train from station $A$ via stations $B$, $C$ to station $D$ as shown in Figure~\ref{fig:example1}.
E.g. the train departs at station $A$ at 23:05 (hh:mm) and arrives at station $B$ at 0:55 at the next day.
The length of this elementary connection is 1:50 = 110 minutes.
\{(2, 1, $C$, $E$, 3:00, 4:00), (3, 2, $C$, $E$, 4:00, 5:00)\} describe elementary connections of two trains from station $C$ to $E$, the first train departs at 3:00 and arrives at 4:00, the second train one hour later.

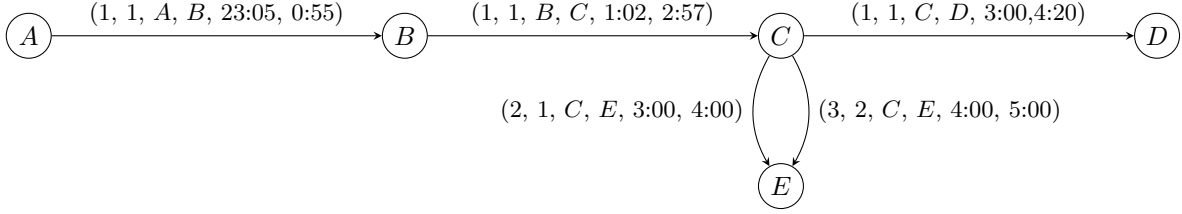
\begin{figure}[htb]
\centering
\tikzstyle{vertex}=[circle,minimum size=17pt,inner sep=0pt,draw=black]
\tikzstyle{weight} = [font=\small]
 
\begin{tikzpicture}[scale=1,auto,swap,>=stealth,->]
    \foreach \pos/\name in {{(0,2)/A}, {(5,2)/B}, {(10,2)/C}, {(15,2)/D}, {(10,0)/E}} \node[vertex] (\name) at \pos {$\name$};
    \foreach \source/ \dest /\weight in {A/B/{(1, 1, $A$, $B$, 23:05, 0:55)}, B/C/{(1, 1, $B$, $C$, 1:02, 2:57)}, C/D/{(1, 1, $C$, $D$, 3:00,4:20)}} \path (\source) edge node[weight,above] {\weight} (\dest);
    \path[bend right] (C) edge node[weight,left] {(2, 1, $C$, $E$, 3:00, 4:00)} (E);
    \path[bend left] (C) edge node[weight,right] {(3, 2, $C$, $E$, 4:00, 5:00)} (E);
\end{tikzpicture}
\caption{Every node is a station and every edge an elementary connection.}
\label{fig:example1}
\end{figure}
\end{example}

Let $P = (c_1,\dots,c_k)$ be a sequence of elementary connections together with departure times $dep_i(P)$ and arrival times $arr_i(P)$ for each elementary connection $c_i$, $1 \le i \le k$.
We assume that the times $dep_i(P)$ and $arr_i(P)$ include data regarding also the departure/arrival day by counting time in minutes from the first day of the timetable.
Such a time $t$ is of the form $t = a 1440 + b$, where $a \in [0,N-1]$ and $b\in[0,1439]$.
Hence, the actual time within a day is $t$ (mod 1440) and the actual day is $\floor{t/1440}$.
Such a sequence $P$ is called a \emph{consistent connection} from station $A=S_1(c_1)$ to station $B = S_2(c_k)$ if it fulfills some consistency conditions:
(a) the departure station of $c_{i+1}$ is the arrival station of $c_i$;
(b) the time values $dep_i(P)$ and $arr_i(P)$ correspond to the time values $t_d$ and $t_a$, resp., of the elementary connections (modulo 1440) and respect the transfer times at stations.
More formally, $P$ is a \emph{consistent connection} if the following conditions are satisfied:

\begin{tabular}{rcl}
$c_i$ && is valid on day $\floor{dep_i(P)/1440}$\\
$S_2(c_i)$ & $=$ & $S_1(c_{i+1})$\\
$dep_i(P)$ & $\equiv$ & $t_d(c_i)$ (mod 1440)\\
$arr_i(P)$ & $=$ & $\mathit{length}(c_i) + dep_i(P)$\\
$dep_{i+1}(P) - arr_i(P)$ & $\ge$ & $\begin{cases}0 & \text{if }Z_1(c_{i+1})=Z_2(c_i)\\\mathit{transfer}(S_2(c_i))&\text{otherwise}\end{cases}$
\end{tabular}

\begin{example}

\begin{tabular}{c|c|c|c}
$c_i$ & $(Z_1,Z_2,S_1,S_2,t_d,t_a)$ & $dep_i$ & $arr_i$\\
\hline
$c_1$ & (1, 1, $A$, $B$, 23:05, 0:55) & 23:05 & 24:55 \\
$c_2$ & (1, 1, $B$, $C$,  1:02, 2:57) & 25:02 & 26:57 \\
$c_3$ & (3, 2, $C$, $E$,  4:00, 5:00) & 28:00 & 29:00
\end{tabular}

$P=(c_1,c_2,c_3)$ is a consistent connection with one transfer.
The elementary connections are from Example~\ref{example:elementary_connections}.
Assume a transfer time at station $C$ of $\mathit{transfer}(C) = 5$ minutes.
It would not be consistent to replace $c_3$ with the train that arrives at $arr_3(P)=$ 28:00 since there are only 3 $< 5 = \mathit{transfer}(C)$ minutes between the arrival and the departure at station $C$.
\end{example}

\subsection{Time Query}

A time query $(A,B,t_0)$ consists of a departure station $A$, an arrival station $B$ and a departure time $t_0$ (including the departure day).
Connections are \emph{valid} if they depart not before the given departure time $t_0$.
A time query solves the earliest arrival problem (EAP) by minimizing the difference between the arrival time and the given departure time.

\subsection{Profile Query}

A profile query $(A,B)$ consists of a departure station $A$ and an arrival station $B$.
It computes a dominant set of all consistent connections between $A$ and $B$.

\section{Station Graph Model}\label{s:ordering}

We introduce a model that represents a timetable as a digraph with exactly one node per station.
For a simplified model without transfer times, this is like the time-dependent model.
New is that even with positive transfer times, we keep one node per station and do not have parallel edges.
We got the idea while trying to apply node contraction to time-dependent networks.
Our observation was, that too many shortcuts where added, and there were a lot of edges between the nodes of the same station pair.
In a first step, we tried to reduce the number of nodes by merging route nodes.
We can do this when it is never necessary to transfer between the two routes at this station.
However, we were not successful with this step and eventually came up with the station graph model.

The attribute of an edge $e=(A,B)$ is a set of consistent connections $\fn{e}$ that depart at $A$ and arrive at $B$.
In this section, all connections are consistent, so we omit to mention it again.
Previous models required that all connections of a single edge fulfill the FIFO-property, i.e. they do not overtake each other.
In contrast, we do not require this property and we will see that even for time queries, we can have more than one dominant arrival event per station.

To link two edges $e_1=(A,B)$ and $e_2=(B,C)$ to an edge $e_3=(A,C)$ we need to pairwise link the connections in $\fn{e_1}$ and $\fn{e_2}$.
We need to drop the new connections that are not consistent.
Also some other new connections might not be necessary if they can always be replaced by other new connections without worsening the objective value.
We call such replaceable connections \textit{dominated} in this set of connections and all other connections \textit{dominant}.
Note that the exact definition of domination depends on the objective function.

\subsection{Time Query}
In this section we describe our baseline algorithm to solve the EAP.
We use a Dijkstra-like algorithm on our station graph that stores labels with each station and incrementally corrects them.
A label is a \emph{connection} $P$ stored as a tuple $(Z_1,Z_2,dep,arr)$ where $Z_1$ is the stop event for departure, $Z_2$ is the stop event for arrival and $dep$ and $arr$ are the departure/arrival time including days.
The source station $S_1(P)$ is implicitely given by the query and the target station $S_2(P)$ is implicitly given by the station that stores this connection.

Not all valid connections are relevant for a time query.
We do not need to store connections that can be dominated (replaced) by another stored connection.
Before we can specify when two connections dominate each other, we need some more definitions.
Let $P$ be a connection.
Define $parr(P)$ as the \underline{p}revious \underline{arr}ival time of the stop event $Z_1(P)$ at station $S_1(P)$ or $\perp$ when the train begins there.
And define $ndep(P)$ as the \underline{n}ext \underline{dep}arture time of the stop event $Z_2(P)$ at station $S_2(P)$ or $\perp$ when the train ends there.
When $parr(P) \neq \perp$ then we call $res_d(P) \Def dep(P) - parr(P)$ the \emph{residence time at departure}.
Consequently, when $ndep(P) \neq \perp$ then we call $res_a(P) \Def ndep(P)- arr(P)$ the \emph{residence time at arrival}.
We call it a \emph{critical departure} when $parr(P) \neq \perp$ and $res_d(P) < \mathit{transfer}(S_1(P))$ and a \emph{critical arrival} when $ndep(P) \neq \perp$ and $res_a(P) < \mathit{transfer}(S_2(P))$.

A connection $P$ dominates a connection $Q$ iff all of the following conditions are fulfilled:
\begin{enumerate}
 \item[(1)] $S_1(P) = S_1(Q)$ and $S_2(P) = S_2(Q)$
 \item[(2)] $dep(Q) \le dep(P)$ and $arr(P) \le arr(Q)$
 \item[(3)] $Z_1(Q) = Z_1(P)$ or $Q$ is not a critical departure or $dep(P) - parr(Q) \ge \mathit{transfer}(S_1(P))$
 \item[(4)] $Z_2(Q) = Z_2(P)$ or $Q$ is not a critical arrival or $ndep(Q) - arr(P) \ge \mathit{transfer}(S_2(P))$
\end{enumerate}

We could also see the query algorithm as a multi-criteria shortest path problem with relaxed Pareto dominance \cite{ms-faatc-07}.
The length of a connection is the cost criterion (2), but it is relaxed by the additional train and transfer information (3),(4).

Given an connection $R=(c_1,\dots,c_k)$, we call a connection $(c_1,\dots,c_{i})$ with $1\le i \le k$ a \emph{prefix} of $R$, a connection $(c_j,\dots,c_k)$ with $1 \le j \le k$ a \emph{suffix} of $R$ and a connection $(c_i,\dots,c_j)$ with $1\le i \le j \le k$ a \emph{subconnection} of $R$.

Lemma~\ref{lemma:dominate_replace} formalizes the intuition of the domination relation.
\begin{lemma}
\label{lemma:dominate_replace}
A consistent connection $P$ dominates a consistent connection $Q$ iff for all consistent connections $R$ with subconnection $Q$, we can replace $Q$ by $P$ to get a consistent connection $R'$ with $dep(R) \le dep(R') \le arr(R') \le arr(R)$.
\end{lemma}
\begin{proof}
$\Rightarrow$ $P$ dominates $Q$:
Condition (1) locates $P$ and $Q$ at the same stations.
Condition (2) ensures that we can replace $R=Q$ by $P$ directly or when transfer times are irrelevant.
The last two conditions (3) and (4) ensure that we can replace $Q$ by $P$ even when $Q$ is just a part of a bigger connection and we need to consider transfer times.
The prefix of this bigger connection \wrt $Q$ may arrive in $S_1(Q)$ with stop event $Z_1(Q)$ and condition (3) ensures that it is consistent to transfer to $Z_1(P)$.
Consequently the suffix of this bigger connection \wrt $Q$ may depart in $S_2(Q)$ with stop event $Z_2(Q)$ and condition (4) ensures that it is consistent to transfer to $Z_2(P)$.

$\Leftarrow$ $\forall R$ we can replace $Q$ by $P$:
Condition (1) holds trivially.
We get condition (2) with $R=Q$.
Let $Q=(c_1,\dots,c_k)$ be a critical departure and $Z_1(Q) \neq Z_1(P)$.
Then we get condition (3) when we choose $R$ as the extension of $Q$ that is given by the critical departure.
Let $Q=(c_1,\dots,c_k)$ be a critical arrival and $Z_2(Q) \neq Z_2(P)$.
Then we get condition (4) when we choose $R$ as the extension of $Q$ that is given by the critical arrival.
\end{proof}

With the given model of a connection, we can model waiting times implicitly as the time between the arrival with one train and the departure of the next train.
However initial waiting as it can appear for a time query $(A,B,t_0)$ is currently not possible.
So we introduce an \emph{arrival connection} $P$ that is represented by a tuple $(arr, Z_2)$ and defines a connection from $A$ to $S_2(P)$.
We perform the query as a label correcting algorithm and store a dominant set of arrival connections with a station so that $S_2(P)$ is implicitly given.
An arrival connection is called \emph{consistent} \wrt the query if it is a consistent connection $Q$ with $dep_1(Q) \ge t_0$ and $S_1(Q) = A$.
All arrival connections in this section are consistent \wrt the query so we do not mention it again.
Linking of a set of arrival connections at station $C$ with an edge $(C,D)$ will result in a set of arrival connections at station $D$.

An arrival connection $P$ dominates an arrival connection $Q$ iff all of the following conditions are fulfilled:
\begin{enumerate}
 \item[(1)] $S_2(P) = S_2(Q)$
 \item[(2)] $arr(P) \le arr(Q)$
 \item[(3)] $Z_2(Q) = Z_2(P)$ or $Q$ is not a critical arrival or $ndep(Q) - arr(P) \ge \mathit{transfer}(S_2(P))$
\end{enumerate}

A result of Lemma~\ref{lemma:dominate_replace} is Lemma~\ref{lemma:dominate_replace_eap}.
\begin{lemma}
\label{lemma:dominate_replace_eap}
Let $(A,B,t_0)$ be a time query.
A consistent arrival connection $P$ dominates a consistent arrival connection $Q$ iff for all consistent arrival connections $R$ with prefix $Q$, we can replace $Q$ by $P$ to get a consistent arrival connection $R'$ with $arr(R') \le arr(R)$.
\end{lemma}

To solve the EAP, we manage a set of dominant arrival connections $ac(S)$ for each station $S$.
The initialization of $ac(A)$ at the departure station $A$ is a special case since we have no real connection to station $A$.
That is why we introduce a special stop event $\forall$ and we start with the set $\set{(t_0, \forall)}$ at station $A$.
Our query algorithm then knows that we are able to board all trains that depart at $t_0$ or later.
We perform a label correcting query that uses the minimum arrival time of the (new) connections as key of a priority queue.
This algorithm needs two elementary operations: (1) \emph{link}: We need to traverse an edge $e=(S,T)$ by linking a given set of arrival connections $ac(S)$ with the connections $\fn{e}$ to get a new set of arrival connections to station $T$.
(2) \emph{minimum}: We need to combine the already existing arrival connections at $T$ with the new ones to a dominant set.
These two operations also dominate the runtime of the query algorithm and we describe in Appendix~\ref{sec:implementation} efficient implementations.
The most important part is a suitable order of the connections, primarily ordered by departure time.
The minimum operation is then mainly a linear merge operation and the link operation uses precomputed intervals to look only at a small relevant subset of $\fn{e}$.
We gain additional speedup by combining the link and minimum operation.

We found a solution to the EAP once the key of the priority queue is $\ge$ the minimum arrival time at $B$.

\begin{code}
\Function timeQuery($A$, $B$, $t_0$)\+\\
  \ForEach $S \in \mathcal{B} \setminus A$ \Do $ac(S) \Def \emptyset$\\
  $ac(A) \Def \set{(t_0, \forall)}$\\
  $Q.\text{insert}(t_0, A)$\\
  \While $Q \neq \emptyset$\+\\
    $(t, S) \Def Q$.deleteMin()\\
    \If $S=B$ \Then \Return $t$\\
    \ForEach edge $e \Def (S,T)$\+\\
      $N \Def \text{minimum}\set{ac(T),e.\text{link}(ac(S))}$\\
      \If $N \neq ac(T)$\+\\
        $ac(T) \Def N$\\
        $k \Def \min_{P \in N}arr(P)$\\
        \If $T$ in $Q$ \Then $Q.\text{decreaseKey}(k, T)$ \Else $Q.\text{insert}(k, T)$\-\-\-\\
  \Return $\perp$
\end{code}

\begin{theorem}
The time query in the station graph model solves the EAP.
\end{theorem}
\begin{proof}
The query algorithm only creates consistent connections because link and minimum do so.
Lemma~\ref{lemma:dominate_replace_eap} ensures that there is never a connection with earlier arrival time.
The connections depart from station $A$ not before $t_0$ by initialization.
And since the length of any connection is non-negative, the minimum length arrival connection $P$ at station $B$ after the first deleteMin() of $B$ is a solution to the EAP.
\end{proof}

\subsection{Profile Query}
A profile query $(A,B)$ is similar to a time query.
However, we compute dominant connections instead of dominant arrival connections.
Also we cannot just stop the search when we remove $B$ from the priority queue for the first time.
We are only allowed to stop the search when we know that we have a dominant set of all consistent connections between $A$ and $B$.
E.g. for daily operating trains, we can compute a maximum duration for a set of connections and can use it to prune the search.
The efficient implementations of the minimum and link operation, see Appendix~\ref{sec:implementation}, are also more complex.
Similar to a time query, we use a suitable order of the connections, primarily ordered by departure time.
The minimum operation is an almost linear merge: we merge the connections in descending order and remove dominated ones.
This is done with a sweep buffer that keeps all previous dominant connections that are relevant for the current departure time.
The link operation, it links connections from station $A$ to $S$ with connections from station $S$ to $T$, is more complex: in a nutshell, we process the sorted connections from $A$ to $S$ one by one, compute a relevant interval of connections from $S$ to $T$ as for the time query, and remove dominated connections using a sweep buffer like for the minimum operation.

\section{Contraction}\label{s:contraction}
In this section, we describe how we contract a station graph timetable network.
Interestingly, we cannot directly use the algorithms used for time-dependent road networks \cite{bdsv-tdch-09}.
The most time consuming part of the contraction is the witness search: given a node $v$ and an incoming edge $(u,v)$ and an outgoing edge $(v,w)$, is a shortcut between $u$ and $w$ necessary when we contract $v$?
For time-dependent road networks, a min-max search from $u$ to $w$ is performed.
It computes a small corridor consisting only of the min. and max. cost path; this corridor is used by a profile search to decide the necessity of a shortcut.
However, in timetable networks, min-max search is not successful, especially the maximum travel time for an edge is very high, e.g. when there is no service during the night.
So we need another way to improve the computation of shortcuts.
Our idea is to use a one-to-many search from $u$ and use it to identify necessary shortcuts for the contraction of all its neighbors $v$.
That is an improvement over \cite{gssd-chfsh-08}, where we performed a one-to-many search from $u$ for each of its neighbors separately.
Before we start contracting nodes, we do this for all nodes $u$ and store the necessary shortcuts $(u,w)$ with each node $v$.
However, we do not store the connections of the shortcuts, but compute them at the contraction of $v$.
These stored shortcuts do not take a lot of space since timetable networks are much smaller than road networks.
When we contract a node $v$, we just add the necessary shortcuts after we computed their connections.
We add shortcuts even in case of equality, so only the endpoints of the newly introduced shortcuts are affected and we need to update their stored shortcuts.
Consider a shortcut $(u,w)$ that is added during the contraction of $v$.
Then we currently do not know which pair of edges $(u,w)$, $(w,x)$ needs a shortcut when $w$ is contracted.
A one-to-many forward search from $u$ will give all necessary information to update the stored shortcuts at $w$.
And analogously, a one-to-many backward search from $w$ will give all necessary information to update the stored shortcuts at $u$.
So at most one one-to-many forward search and one one-to-many backward search from each neighbor of $v$ is necessary.
When we add a new shortcut $(u,w)$ but there is already an edge $(u,w)$, we merge both edges so there are never parallel edges.
Avoiding these parallel edges is important for the contraction, as it performs worse on dense graphs.
Thereby, we also ensure that we can uniquely identify an edge with its endpoints.
The one-to-many searches can be limited by the duration of the longest shortcut between $u$ and $w$.
We also limit the number of hops and the number of transfers.
As observed in \cite{gssd-chfsh-08}, this accelerates the witness search at the cost of missed witnesses and potentially more shortcuts.

We could omit loops in static and time-dependent road networks.
But for station graph timetable networks, loops are sometimes necessary.
We will give an example:

\begin{example}
Consider a train that departs at station $A$ at 12:00 then arrives/departs at station $B$ at 12:01, then $C$ at 12:02, then again $B$ at 12:03 and then $D$ at 12:04.
The set of elementary connections is \{(1, 1, $A$, $B$, 12:00, 12:01), (1, 1, $B$, $C$, 12:01, 12:02), (1, 2, $C$, $B$, 12:02, 12:03), (2, 1, $B$, $D$, 12:03, 12:04)\}, also shown in Figure~\ref{fig:example3}.
Let the transfer time at station $B$ be 5 minutes.
We want to go from $A$ to $D$.
In our model, it is not possible to transfer from the train at $B$ at 12:01 to the same train at $B$ at 12:03.
However, it is possible to stay at the train via $C$ and then we get a consistent connection from $A$ to $D$ arriving at 12:04.
Thus when we contract station $C$, we need to add a loop at station $B$.

\begin{figure}[htb]
\centering
\tikzstyle{vertex}=[circle,minimum size=17pt,inner sep=0pt,draw=black]
\tikzstyle{weight} = [font=\small]
 
\begin{tikzpicture}[scale=1,auto,swap,>=stealth,->]
    \foreach \pos/\name in {{(0,2)/A}, {(5,2)/B}, {(10,2)/C}, {(5,0)/D}} \node[vertex] (\name) at \pos {$\name$};
    \path (A) edge node[weight,above] {(1, 1, $A$, $B$, 12:00, 12:01)} (B);
    \path[bend left=10] (B) edge node[weight,above] {(1, 1, $B$, $C$, 12:01, 12:02)} (C);
    \path[bend left=10] (C) edge node[weight,below] {(1, 2, $C$, $B$, 12:02, 12:03)} (B);
    \path (B) edge node[weight,left] {(2, 1, $B$, $D$, 12:03, 12:04)} (D);
\end{tikzpicture}
\caption{Every node is a station and every edge an elementary connection.
Let $\mathit{transfer}(B) = 5$, then a loop at station $B$ is necessary after the contraction of station $C$.}
\label{fig:example3}
\end{figure}
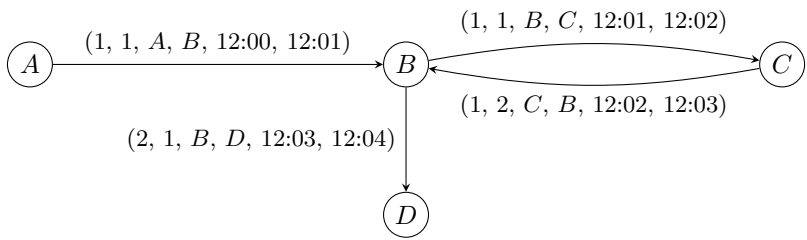
\end{example}

Loops also make the witness computation and the update of the stored shortcuts more complex.
A shortcut $(u,w)$ for node $v$ with loop $(v,v)$ must not only represent the path $\seq{u,v,w}$ but also $\seq{u,v,v,w}$.
So when we add a shortcut $(v,v)$ during the contraction of another node, we need to recompute all stored shortcuts of node $v$.

We use two terms for the node ordering:
(a) The edge quotient, the quotient between the amount of shortcuts added and the amount of edge removed from the remaining graph.
(b) The hierarchy depth, an upper bound on the amount of hops that can be performed in the resulting hierarchy.
Initially, we set depth($u$) = 0 and when a node $v$ is contracted, we set depth($u$) = max(depth($u$),depth($v$)+1) for all neighbors $u$.
We weight (a) with 10 and (b) with 1 in a linear combination to compute the node priorities.
The nodes are contracted by computing independent node sets with with a 2-neighborhood as described in \cite{v-ptdch-09}.

\section{Query}

As any hierarchical time-dependent time query, we suffer from the unknown arrival time.
This problem is solved in \cite{bdsv-tdch-09} by computing a corridor from the target node using a min-max backward search in the hierarchy.
We replace this min-max search by a breath first search since min-max search does not work so well.
Also we do not use the stall-on-demand technique since the additional overhead does not pay off.
Our forward and backward search are not interleaved; we first perform the backward and then the forward search.
In road networks \cite{bdsv-tdch-09}, even for the profile query, the computation of a forward and backward corridor using min-max searches is beneficial.
Again, we omit the min-max searches and just use a bidirectional interleaved profile query.

\begin{theorem}
The time query in a station graph CH solves the EAP.
\end{theorem}
\begin{proof}
The proof is similar to the one given in \cite{g-ch-08}.
Here, we will only present a short outline of the proof.
From a given time query $(A,B,t_0)$ and a shortest path in the original graph, we construct a path that is found by our query algorithm.
Since we use profile searches for the witness computation, Lemma~\ref{lemma:dominate_replace} ensures that we can recursively remove a node from the path by replacing the neighboring edges either with a shortcut or a witness path.
Different to \cite{g-ch-08}, a shortest path may contain a station more than once.
Thus when we remove a node with loop, we remove not only two but three edges from the path.
We cover this case with our special treatment of nodes with loops during the contraction.
\end{proof}

\section{Experiments}\label{s:experiments}

\paragraph{Environment.}
Experiments have been done on one core of a dual Xeon 5345 processor clocked at 2.33 GHz with 16 GB main memory and $2\times2\times$ 4~MB of cache, running SuSE Linux 11.1 (kernel 2.6.27).
The program was compiled by the GNU C++ compiler 4.3.2 using optimization level 3.

\paragraph{Test Instances.}
We have used real-world data from the European railways.
The network of the long distance connections of Europe (\texttt{eur-longdist}) is from the winter period 1996/97.
The network of the local traffic in Berlin/Brandenburg (\texttt{ger-local1}) and of the Rhein/Main region in Germany (\texttt{ger-local2}) are from the winter period 2000/01.
The sizes of all networks are listed in Table~\ref{tab:network_sizes}.

\begin{table}[htb]
\centering
\caption{Network sizes.
We give the the number of nodes and edges in the resulting graph for both models.
Note that in the station graph model, the number of nodes corresponds to the number of stations.}
\label{tab:network_sizes}
\begin{tabular}{l|r|r|r|r|r|r}
        &          & trains/& elementary  & \multicolumn{2}{|c|}{time-dependent} & \multicolumn{1}{|c}{station} \\
network & stations & buses & connections & nodes & edges  & edges \\
\hline
\texttt{eur-longdist} & 30\,517 & 167\,299 & 1\,669\,666 & 535\,963 & 1\,456\,904 & 88\,091 \\
\texttt{ger-local1}   & 12\,069 &  33\,227 &    680\,176 & 225\,797 &    600\,690 & 33\,473 \\
\texttt{ger-local2}   &  9\,902 &  60\,889 & 1\,128\,465 & 183\,207 &    704\,673 & 26\,678 \\
\end{tabular}
\end{table}

\paragraph{Results.}
We selected 1\,000 random queries and give average performance measures.
We compare the time-dependent model and our new station model using a simple unidirectional Dijkstra algorithm in Table~\ref{tab:performance_model}.
Time queries have a good query time speedup above 4 and even more when compared to the \#delete mins.
However, since we do more work per delete min, this difference is expected.
Profile queries have very good speedup around 6 to 8 for all tested instances.
Interestingly, our speedup of the \#delete mins is even better than for time queries.
We assume that more re-visits occur since there are often parallel edges between a pair of stations represented by its route nodes.
Our model does not have this problem since we have no parallel edges and each station is represented by just one node.
It is not possible to compare the space consumption per node since the number of nodes is different in the different models.
So we give the the absolute memory footprint: it is so small, we did not even try to reduce it, altough there is potential for that.

\begin{table}[htb]
\centering
\caption{Performance of the station graph model compared to the time-dependent model on plain Dijkstra queries.
\emph{\#delete mins} denotes the number of nodes removed from the priority queue, query \emph{times} are given in milliseconds.
Moreover, we report the \emph{speedup} over the corresponding value from the time-dependent model.}
\label{tab:performance_model}
\begin{tabular}{l|l|r|rrrr|rrrr}
        &       &       & \multicolumn{4}{|c|}{\textsc{time-queries}} & \multicolumn{4}{|c}{\textsc{profile-queries}}\\
        &       & space & \#delete & speed & time & speed & \#delete & speed & time & speed \\
network & model &  [MB] &  mins    & up    & [ms] & up    & mins    & up    & [ms] & up \\
\hline
\texttt{eur-longdist} & Dependent & 27.5 & 253\,349 & 1.0 & 68.0 & 1.0 & 1\,900\,520 & 1.0 & 2\,482 & 1.0 \\ 
                      & Station   & 48.3 &  14\,519 &17.4 & 14.3 & 4.8 &     48\,420 &39.3 &    333 & 7.5 \\ 
\hline
\texttt{ger-local1}   & Dependent & 11.6 & 110\,864 & 1.0 & 26.0 & 1.0 & 1\,172\,320 & 1.0 & 1\,642 & 1.0 \\ 
                      & Station   & 19.6 &   5\,977 &18.5 &  6.2 & 4.2 &     33\,647 &34.8 &    289 & 5.7 \\ 
\hline
\texttt{ger-local2}   & Dependent & 16.1 &  98\,684 & 1.0 & 25.6 & 1.0 & 1\,145\,560 & 1.0 & 2\,822 & 1.0 \\ 
                      & Station   & 29.4 &   5\,097 &19.4 &  5.5 & 4.7 &     27\,701 &41.3 &    345 & 8.2 \\ 
\end{tabular}
\end{table}

Before we present our results for CH, we would like to mention that we were unable to contract the same networks in the time-dependent model.
The contraction took days and the average degree in the remaining graph exploded.
Even when we contracted whole stations with all of its route nodes at once, it did not work.
It failed since the necessary shortcuts between all the onboard nodes multiplied quickly.
So we developed the station graph model to fix these problems.
Table~\ref{tab:performance_ch} shows the the resulting preprocessing and query performance.
We get preprocessing times between 4 to 6 minutes using a hop limit of 7, these times are exceptional low compared to previous publications \cite{d-tdsr-09,bdgm-atdmc-09}.
This is sufficient to reduce time queries below 1\,ms for all tested instances.
CHs work very well for \texttt{eur-longdist} where we get speedups of more than 35 for time queries and 54 for profile queries.
When we multiply the speedup of the time-dependent model, we get even a speedup of 172 (time) and 409 (profile) respectively.
The network \texttt{ger-local2} is also suited for CH, the ratio between elementary connections and stations is just very high so there is more work per settled node.
More difficult is \texttt{ger-local1}; in our opinion, this network is less hierarchically structured.
We see that on the effect of different hop limits for precomputation: we chose 7 as a hop limit for fast preprocessing and then selected 18 as it achieves the best query times for \texttt{ger-local1}.
The smaller hop limit increases time query times by about 50\%, whereas the other two networks just suffer an increase of about 20\%.
So important witnesses in \texttt{ger-local1} contain more edges indicating a lack of hierarchy.

We do not really have to worry about preprocessing space since those networks are very small.
The number of edges roughly doubles for all instances.
We observe similar results for static road networks \cite{gssd-chfsh-08}, but there we can save space with bidirectional edges.
But in timetable networks, we do not have bidirectional edges with same weight, so we need to store them separately.
In contrast to time-dependent road networks \cite{bdsv-tdch-09} (Germany midweek: 0.4 GB $\rightarrow$ 4.4 GB), CHs increase the memory consumption by not more than a factor 2.8 (\texttt{ger-local1}: 19.6 MB $\rightarrow$ 54.6 MB).
So in our model, CHs have not only fast preprocessing and query times but are even space efficient.

It is hard to compare ourselves to \cite{bdgm-atdmc-09} since we do not use the same graph and also not the same scenario.
When we just look at the numbers, we are about 2-3 orders of magnitude faster, both for preprocessing and for query times, but our scenario lacks a lot of features.

\begin{table}[htb]
\centering
\caption{Performance of CH.
Preprocessing times are given in seconds, the overhead in megabytes.
Moreover, we report the increase in edge count over the input.
\emph{\#delete mins} denotes the number of nodes removed from the priority queue, query \emph{times} and \emph{speed-up} over a plain Dijkstra are given.
}
\label{tab:performance_ch}
\begin{tabular}{l|r|rrr|rrrr|rrrr}
        &       & \multicolumn{3}{|c|}{\textsc{preprocessing}} & \multicolumn{4}{|c|}{\textsc{time-queries}} & \multicolumn{4}{|c}{\textsc{profile-queries}}\\
        & hop-  & time & space & edge     & \#del. & speed & time & speed & \#del. & speed & time & speed \\
network & limit & [s]  &  [MB] & inc.     & mins     & up    & [$\mu$s] & up    & mins     & up    & [ms] & up \\
\hline
\texttt{eur-}     &  7 &    277 & 48.6 &  88\% & 194 & 75 & 399 & 35.8 & 266 & 182 &  6.1 & 54.6 \\ 
\texttt{longdist} & 18 & 1\,000 & 48.0 &  85\% & 186 & 78 & 331 & 43.2 & 260 & 186 &  5.4 & 61.7 \\ 
\hline
\texttt{ger-}     &  7 &    244 & 35.0 & 134\% & 203 & 29 & 968 &  6.4 & 431 &  78 & 50.2 &  5.8 \\ 
\texttt{local1}   & 18 & 1\,348 & 33.4 & 125\% & 186 & 32 & 658 &  9.5 & 389 &  86 & 40.7 &  7.1 \\ 
\hline
\texttt{ger-}     &  7 &    367 & 41.2 & 122\% & 155 & 33 & 424 & 13.0 & 251 & 110 & 17.8 & 19.4 \\ 
\texttt{local2}   & 18 &    746 & 39.8 & 115\% & 151 & 34 & 358 & 15.4 & 258 & 107 & 15.7 & 22.0 \\ 
\end{tabular}
\end{table}

\section{Conclusions}\label{s:conclusions}
Our model, that has just one node per station, is clearly superior to the time-dependent model for the given scenario.
Although the link and minimum operations are more expensive, we are still faster than in the time-dependent model since we need to execute them less often.
Also all known speedup techniques that work for the time-dependent model should work for our new model.
Most likely, they even work better since the hierarchy of the network is more visible because of the one-to-one mapping of stations to nodes and the lack of parallel edges.
We demonstrate that with CHs, they work very well.
We achieve query times below 1\,ms with preprocessing times around 6 minutes.
Our algorithm is therefore suitable for simple web services, where small query times are very important and can compensate for the quality of the results.

\section{Future Work}\label{s:future_work}
Our tested instances are very small, especially when we compare it to road networks.
However, it is very hard to get more data from the travel agencies.
Moreover, combining data from several agencies is hard because of proprietary formats, different levels of granularity, etc.
The experience for road networks shows that speedups increase with the size of the network.
So we can hope for even better speedups in larger networks.
Not only larger, but more realistic scenarios are also important.
Our scenario is very basic and the one from \cite{bdgm-atdmc-09} has a lot of features.
It would be interesting to see the effects of the single features on the speedup techniques, i.e. where do the 2-3 orders of magnitude in performance come from?
What are the performance costs of multi-criteria, traffic days, footpaths, etc. and what can we learn to develop faster algorithms?
Or is our implementation just highly tuned and the implementation of \cite{bdgm-atdmc-09} could be improved with our ideas?

We already saw in Section~\ref{s:experiments} that CHs work worse when the network is not very hierarchical.
Public transportation in Europe is very good and also has a lot of hierarchy.
But in other countries, e.g. the United States or Brazil, these networks are worse.
There exists e.g. large bus networks with little structure; most likely improved goal directed speedup techniques are necessary for such networks.
In some networks, there is no continuous service throughout the day and stations are only important for a few hours.
Also the station placement is not very good in every network, e.g. there is a separate station for each bus, and transfers often require a lot of walking, or there are dozens of stops in a very small area.
We conjecture, that even in a network with hierarchy, contraction fails when there are too many important footpaths.

But there are not only difficult but also different public transportation timetables, e.g. in London, they specify how frequent a bus will arrive at a station instead of fixed arrival times.
We can integrate this frequency based routing in our station graph model.
Technically, an edge stores not only a set of timetable connections but also a set of frequency intervals.
Of course, in most cases, one of the sets will be empty.
When we link an timetable connection and a frequency based connection, we get a connection with a specific departure and arrival time and thus we can handle it as a normal timetable connection.

We should also evaluate different speedup techniques in combination with the station graph model.
We only tried CHs, since our original motivation was to proof that they work in timetable networks.
But combinations with goal-directed techniques should accelerate the query even more.

\bibliographystyle{plain}
\bibliography{references}

\begin{thebibliography}{10}

\bibitem{bdsv-tdch-09}
Veit Batz, Daniel Delling, Peter Sanders, and Christian Vetter.
\newblock {Time-Dependent Contraction Hierarchies}.
\newblock In {\em Proceedings of the 11th Workshop on Algorithm Engineering and
  Experiments (ALENEX'09)}, pages 97--105. SIAM, April 2009.

\bibitem{bdgm-atdmc-09}
Annabell Berger, Daniel Delling, Andreas Gebhardt, and Matthias
  {M{\"u}ller--Hannemann}.
\newblock {Accelerating Time-Dependent Multi-Criteria Timetable Information is
  Harder Than Expected}.
\newblock In {\em Proceedings of the 9th Workshop on Algorithmic Approaches for
  Transportation Modeling, Optimization, and Systems (ATMOS'09)}, Dagstuhl
  Seminar Proceedings, 2009.
\newblock To appear.

\bibitem{bm-somcs-09}
Annabell Berger and Matthias {M{\"u}ller--Hannemann}.
\newblock {Subpath-Optimality of Multi-Criteria Shortest Paths in Time- and
  Event-Dependent Networks}.
\newblock Technical Report~1, University Halle-Wittenberg, Institute of
  Computer Science, 2009.

\bibitem{bj-tnmaf-04}
Gerth Brodal and Riko Jacob.
\newblock {Time-dependent Networks as Models to Achieve Fast Exact Time-table
  Queries}.
\newblock In {\em Proceedings of ATMOS Workshop 2003}, pages 3--15, 2004.

\bibitem{d-tdsr-09}
Daniel Delling.
\newblock {Time-Dependent SHARC-Routing}.
\newblock {\em Algorithmica}, July 2009.
\newblock Special Issue: European Symposium on Algorithms 2008.

\bibitem{dw-tdrp-09}
Daniel Delling and Dorothea Wagner.
\newblock {Time-Dependent Route Planning}.
\newblock In Ravindra~K. Ahuja, Rolf~H. M{\"o}hring, and Christos Zaroliagis,
  editors, {\em {Robust and Online Large-Scale Optimization}}, Lecture Notes in
  Computer Science. Springer, 2009.
\newblock Accepted for publication, to appear.

\bibitem{g-ch-08}
Robert Geisberger.
\newblock {Contraction Hierarchies}.
\newblock Master's thesis, Universit{\"a}t Karlsruhe (TH), Fakult{\"a}t f{\"u}r
  Informatik, 2008.

\bibitem{gssd-chfsh-08}
Robert Geisberger, Peter Sanders, Dominik Schultes, and Daniel Delling.
\newblock {Contraction Hierarchies: Faster and Simpler Hierarchical Routing in
  Road Networks}.
\newblock In Catherine~C. McGeoch, editor, {\em Proceedings of the 7th Workshop
  on Experimental Algorithms (WEA'08)}, volume 5038 of {\em Lecture Notes in
  Computer Science}, pages 319--333. Springer, June 2008.

\bibitem{mn-eatm-98}
Patrice Marcotte and Sang Nguyen, editors.
\newblock {\em {Equilibrium and Advanced Transportation Modelling}}.
\newblock Kluwer Academic Publishers Group, 1998.

\bibitem{ms-faatc-07}
Matthias {M{\"u}ller--Hannemann} and Mathias Schnee.
\newblock {Finding All Attractive Train Connections by Multi-Criteria Pareto
  Search}.
\newblock In {\em {Algorithmic Methods for Railway Optimization}}, volume 4359
  of {\em Lecture Notes in Computer Science}, pages 246--263. Springer, 2007.

\bibitem{mw-pspof-01}
Matthias {M{\"u}ller--Hannemann} and Karsten Weihe.
\newblock {Pareto Shortest Paths is Often Feasible in Practice}.
\newblock In {\em Proceedings of the 5th International Workshop on Algorithm
  Engineering (WAE'01)}, volume 2141 of {\em Lecture Notes in Computer
  Science}, pages 185--197. Springer, 2001.

\bibitem{n-t-95}
Karl Nachtigall.
\newblock {Time depending shortest-path problems with applications to railway
  networks}.
\newblock {\em European Journal of Operational Research}, 83(1):154--166, 1995.

\bibitem{or-spmda-90}
Ariel Orda and Raphael Rom.
\newblock {Shortest-Path and Minimum Delay Algorithms in Networks with
  Time-Dependent Edge-Length}.
\newblock {\em Journal of the ACM}, 37(3):607--625, 1990.

\bibitem{or-mwptd-91}
Ariel Orda and Raphael Rom.
\newblock {{Minimum Weight Paths in Time-Dependent Networks}}.
\newblock {\em Networks}, 21:295--319, 1991.

\bibitem{pswz-emtip-07}
Evangelia Pyrga, Frank Schulz, Dorothea Wagner, and Christos Zaroliagis.
\newblock {Efficient Models for Timetable Information in Public Transportation
  Systems}.
\newblock {\em ACM Journal of Experimental Algorithmics}, 12:Article 2.4, 2007.

\bibitem{sww-daola-00}
Frank Schulz, Dorothea Wagner, and Karsten Weihe.
\newblock {Dijkstra's Algorithm On-Line: An Empirical Case Study from Public
  Railroad Transport}.
\newblock {\em ACM Journal of Experimental Algorithmics}, 5, 2000.

\bibitem{swz-umlgt-02}
Frank Schulz, Dorothea Wagner, and Christos Zaroliagis.
\newblock {Using Multi-Level Graphs for Timetable Information in Railway
  Systems}.
\newblock In {\em Proceedings of the 4th Workshop on Algorithm Engineering and
  Experiments (ALENEX'02)}, volume 2409 of {\em Lecture Notes in Computer
  Science}, pages 43--59. Springer, 2002.

\bibitem{v-ptdch-09}
Christian Vetter.
\newblock {Parallel Time-Dependent Contraction Hierarchies}.
\newblock Master's thesis, Universit{\"a}t Karlsruhe (TH), Fakult{\"a}t f{\"u}r
  Informatik, 2009.

\end{thebibliography}
{\clearpage
\appendix

\section{Implementation Details}
\label{sec:implementation}

The set of connections for each edge are stored as an ordered array.
They are primarily ordered by departure time and then secondarily ordered by arrival time and then critical arrival before non-critical arrivals.
For each connection, we only store a representative with departure time in $[0,1439]$.
So the array actually represents a larger \emph{outrolled} array by simply concatenating and shifting the times by 1440.
The first piece is at day 0, the second piece at day 1, etc.
The set of arrival connections is also stored by an ordered array.
They are primarily ordered by arrival time and then critical arrival before non-critical arrivals.
This ensures that no arrival connection in the array dominates an arrival connection with lower index.

A basic link algorithm for the time query would link to all connections and afterwards remove the dominated arrival connections.
Let $g \Def \abs{ac(S)}$ and $h \Def \abs{\fn{e=(S,T)}}$, the basic algorithm would create up to $\Th{g \cdot h}$ connections.
Especially $h$ can be very large even though usually only a small range in $\fn{e}$ is relevant for the link operation.
When we identified the first connection we can link to, we can store a \emph{dominant range} with it so that all connections after this range result in dominated arrival connections.
We could distinguish between linking to a certain connection with and without transfer but we restrict ourselves only to the case with transfer.
This results in a practically very efficient link operation.
So given an array of arrival connections $ac(S)$ and an array of connections of an edge $\fn{e}$ to relax, the link will work as follows:

\begin{enumerate}
  \item $edt \Def \min_{P \in ac(S)}arr(P)$ (mod 1440) // \underline{e}arliest \underline{d}eparture \underline{t}ime, in $[0,1439]$
  \item $ett \Def edt + \textit{transfer}(S)$ (mod 1440) // \underline{e}arliest departure with \underline{t}ransfer \underline{t}ime
  \item Find first connection $P_n\in \fn{e}$ with minimal $\textit{cycle difference}(edt, dep(P_n))$ using buckets.
  \item Find first connection $P_t\in \fn{e}$ with minimal $\textit{cycle difference}(ett, dep(P_t))$.
Connection $P_t$ gives a dominant range that is identified by the first connection $P_e$ outside the range
This partitions the outrolled array of $\fn{e}$:

{\unitlength1.2em\begin{picture}(20,2)
\put(0,0){\framebox(2,1){...}}
\put(2,1){\makebox(1,1){$P_n$}}
\put(2,0){\framebox(7,1){link w/o transfer}}
\put(9,1){\makebox(1,1){$P_t$}}
\put(9,0){\framebox(6,1){link w/ transfer}}
\put(15,1){\makebox(1,1){$P_e$}}
\put(15,0){\framebox(2,1){...}}
\end{picture}}\\
We may only link to a connection in $[P_n,P_t)$ without (w/o) transfers and thus all arrival connections in $ac(S)$ are relevant to decide which consistent arrival connections we can create there.
It is consistent to link to all connections with transfers from $P_t$ on.
\item While we link, we remember the minimal arrival time and use it to skip dominated arrival connections.
\item Finally we sort the resulting connections and remove the dominated ones.
This step is necessary because the minimum arrival time may decrease while we link and we may have to remove duplicates, too.
\end{enumerate}

Given two sets of arrival connections at a node, we want to build the dominant set of the union, the \emph{minimum}.
This can be done in linear time by just merging them.
Sometimes arrival connections are equivalent but not identical.
Two arrival connections are \emph{equivalent} if they are identical or have the same arrival time and neither of them has a critical arrival.
In this case we must keep just one of them.
We base this decision so that we minimize the number of queue inserts in the query algorithm, e.g. prefer the one from $ac(T)$ if available.

\subparagraph{Runtime.}
The above link operation is is more complex than a usual link operation that maps departure time to arrival time.
However, we give an idea why this link operation is very fast and may work in constant time in many cases.
The experiments in Section~\ref{s:experiments} show that it is indeed very efficient.
Let $b$ be the number of connections in the bucket.
Let $c_d$ be the number of connections that depart within the transfer time window $[P_n,P_t)$ at the station.
Let $c_a$ be the number of arrival connections $|ac(S)|$.
Let $r$ be the number of connections that depart within the range $[P_t, P_e)$.
The runtime of link is then $\Oh{b+c_dc_a+r}$.
We choose the number of buckets proportional to the number of connections, so $b$ is in many cases constant.
For linking connections without transfer, we have the product $\Oh{c_dc_a}$ as summand in the runtime.
We could improve the product down to $\Oh{c_d + c_a + u}$ with hashing, where $u$ is the number of linked connections.
But this is slower in practice since $c_d$ and $c_a$ are usually very small.
That is because the station-dependent transfer time window is usually very small, and also only very few connections depart and arrive within a single window.
It is harder to give a feeling of the size of the range $[P_t, P_e)$.
Assume that every connection operates daily.
Let be $d$ the difference between the length of $P_t$ and the minimum length of any connection in $\fn{e}$.
$d+\mathit{transfer}(S)$ is an upper bound on the size of the time window of this range.
So when $d$ is small, and this should be in many cases, also $r$ is small.

\paragraph{Computing the Dominant Range.} Besides the buckets we also need to compute the dominant ranges.
Lemma~\ref{l:dominant_range} gives the instructions how to efficiently compute them using a sweep algorithm approach.

\begin{lemma}
\label{l:dominant_range}
Given an array $F$ of connections between two stations $S_1$ and $S_2$ that operate daily. The array is primarily orderd by departure time and the secondarily ordered by arrival time and then critical arrival before non-critical arrivals.
Let $P$ be a arrival connection at station $S_1$ that can link with a connection $Q$ in the outrolled array with a transfer.
Then all connections that are later in this outrolled array and may not be dominated by the new arrival connection, created by the link of $P$ and $Q$, depart earlier than $dep(Q)+d+\mathit{transfer}(S_2)$.
$d$ is the difference between the length of $Q$ and the minimum length of any connection in $F$.
\end{lemma}
\begin{proof}
Let $Q'$ be a connection that does not depart earlier than $dep(Q)+d+\mathit{transfer}(S_2)$.
Since $arr(P) + \mathit{transfer}(S_1) \le dep(Q)$, and $dep(Q) \le dep(Q')$, we can link $P$ with $Q'$.
By definition of $d$ is $\mathit{length}(Q) \le \mathit{length}(Q')+d$ and thus $arr(Q) = dep(Q) + \mathit{length}(Q) \le dep(Q) + \mathit{length}(Q')+ d \le (dep(Q')-d-\mathit{transfer}(S_2)) + \mathit{length}(Q') + d \le arr(Q') - \mathit{transfer}(S_2)$.
When $Q'$ has a critical arrival, then $arr(Q') \le ndep(Q')$ so that $P$ linked with $Q$ will dominate $P$ linked with $Q'$.
\end{proof}

\paragraph{Linking two edges} for shortcuts and profile search is done by doing the dominant range computation at link time.
We change the order of the connections in the array of connections when supporting this operation.
They are still primarily ordered by departure.
But within the same departure time, the dominant connection should be after the dominated one.
That allows for an efficient backward sweep to remove dominated connections.
Namely we secondarily order by length descending, then non-critical before critical departure, then non-critical before critical arrival.
Finally, we order by the first and last stop event, preferring a stop event with critical departure or arrival.
The last order is necessary for an efficient building of a dominant union (minimum) of two connection sets where the preference is on one set.


Given two edges $e_1= (S_1,S_2)$ and $e_2 = (S_2,S_3)$, we want to link all consistent connections to create $\fn{e_3}$ for an an edge $e_3 = (S_1,S_3)$.
A trivial algorithm would link each consistent pair of connections in $\fn{e_1}$ and $\fn{e_2}$ and then compare each of the resulting connections with all other connections to find a dominant set of connections. However, this is impractical for large $g=\abs{\fn{e_1}}$ and $h=\abs{\fn{e_2}}$.
We would create $\Th{g\cdot h}$ new connections and do up to $\Th{(gh)^2}$ comparisons.

So we propose a different strategy for linking that is considerably faster for practical instances.
We process the connections in $\fn{e_1}$ in descending order.
Given a connection $P$, we want to find a connection $Q$ that dominates $P$ at the departure at $S_1$.
So we only need to link $P$ to connections in $\fn{e_2}$ that depart in $S_2$ after the arrival of $P$ but before the arrival of $Q$.
Preferably we want to find the $Q$ with the earliest arrival time.
However, we find the $Q$ with the earliest arrival time in $S_2$ with $dep(Q) \ge dep(P) + \mathit{transfer}(S_2)$.
Then $Q$ will not only dominate $P$ at the departure but also any connection departing not later than $P$.
So we can use a simple finger search to find $Q$.
Now we link $P$ only to connections in $\fn{e_2}$ departing between the arrival of $P$ and $Q$.
We use finger search to find the first connection that departs in $\fn{e_2}$ after the arrival of $P$.
Of course, we need to take the transfer time at $S_2$ into account when we link.
It is not always necessary to link to all connections that depart before $Q$ arrives; we can use the knowledge of the minimum length in $\fn{e_2}$ to stop linking when we cannot expect any new dominant connections.
The newly linked connections may (1) not be dominant and also may (2) not be in order.

(1) To remove dominated connections, we use a sweep buffer that has as state the current departure time and holds all relevant connections with higher order to dominate a connection with the current departure time.
The number of relevant connections is usually small.
We need at most all the connections that depart less than $\mathit{transfer}(S_1)$ later than the current departure time and also all connections that depart at least $\mathit{transfer}(S_1)$ later than the current departure time but their arrival time is not more than $\mathit{transfer}(S_3)$ later than the current earliest arrival time.
Assuming that only few connections depart in $S_1$ within $\mathit{transfer}(S_1)$ minutes, and only few connections arrive in $S_3$ within $\mathit{transfer}(S_3)$ minutes, the sweep buffer has only a few entries.

(2) Connections can only be unordered within a range with same departure time, e.g. when they have ascending durations.
So we use the idea of insertion sort to reposition a connection that is not in order.
While we reposition a new connection, we must check whether it dominates the connections that are now positioned before it.
E.g. a new connection with same departure than the previous one but smaller duration may dominate the previous one if the departure is not critical.

After we processed all connections in $\fn{e_1}$, we have a superset of $\fn{e_3}$ that is already ordered, but some connections may be dominated.
This happens when the dominant connection departs after midnight and the dominated connection before, so the periodic border is between them.
To remove all dominated connections, we continue scanning backwards through the new connections but now on ``day -1'' using the sweep buffer.
We can stop when no connection in the sweep buffer is of ``day 0''.

\subparagraph{Runtime.}
We give an idea why this link operation is very fast and may work in linear time in many cases.
The experiments in Section~\ref{s:experiments} show that it is indeed very efficient.
Let $c_P$ be the size of the range in $\fn{e_3}$ that departs between the arrival of $P$ and $Q$.
Let $b_P$ be the runtime of the finger search to find the erliest connection in $\fn{e_2}$ that departs after the arrival of $P$
Let $s$ be the maximum number of relevant connections in the sweep buffer.
The runtime of link is then $\Oh{\sum_{P \in \fn{e_1}}{(c_Ps + b_P)}}$.
This upper bound reflects the linking and usage of the sweep buffer.
The backward scanning on ``day -1'' is also included, since it just adds a constant factor to the runtime.
The finger search for $Q$ is amortized in $\Oh{1}$, so it is also included in the runtime above.
It is hard to get a feeling for $c_P$ and $b_P$, they can be large when $h=\abs{\fn{e_2}}$ is much larger than $g=\abs{\fn{e_1}}$.
Under the practical assumption that $\sum_{P \in \fn{e_1}}{\left(c_P + b_P\right)} = \Oh{g+h}$, we get a runtime of $\Oh{(g+h)s}$.
As we already argued when we described the sweep buffer, $s$ is small and in many cases constant, so our runtime should be $\Oh{g+h}$ in many cases.

\paragraph{Building the Minimum of two Sets of Connections.}
Query algorithms need two basic operations: link and minimum.
For newly visited nodes only link is relevant, otherwise a minimum follows a link so it is efficient to integrate both.
But first we will describe a standalone minimum operation, we use it compare witness paths and possible shortcuts.
It is basically a backwards merge of the ordered arrays of connections and uses a sweep buffer as for the link operation.
Also like the link operation, we continue backward scanning on ``day -1'' to get rid of dominated connections over the periodic border.

Like for an arrival connection, two connections are \emph{equivalent} when they have the same length, an equivalent departure and equivalent arrival.
Two connections $P$ and $Q$ have an \emph{equivalent departure} when their departure is identical or when the departure is not critical and they have the same departure time.
Analogously, two connections $P$ and $Q$ have an \emph{equivalent arrival} when their arrival is identical or when the arrival is not critical and they have the same arrival time.
Because of the order, equivalent connections are next to each other.
So we can easily detect them during the merge.
Tie breaking is done in a way to reduce the number of priority queue operations.

\subparagraph{Runtime.}
Let $g$ and $h$ be the cardinalities of the two sets we merge.
Let $s$ be the maximum size of the sweep buffer.
Then the runtime of the minimum operation is $\Oh{(g+h)s}$.
Since $s$ is small and in many cases constant, the runtime should be $\Oh{g+h}$ in many cases.

\paragraph{Integrating Link and Minimum.}
A minimum operation always follows a link operation when we relax an edge to an already reached station $S$.
This happens quite often for profile queries, so we can exploit this to tune our algorithm.
It is quite simple, we directly process the newly linked connections one by one and directly merge them with the current connections at $S$.
When a new connection is not in order, we fix this with the insertion sort idea.
The rest is like in the stand-alone minimum operation.
This integration reduces required memory allocations and gives significant speedups.

\paragraph{More Features.}
Our implementations of \emph{link} and \emph{minimum} are tuned for the EAP.
But we believe that we can extend our implementation to more features, e.g. multi-criteria, without loosing too much efficiency.
The idea with the dominant ranges can be generalized to support different dominance relations.
Also, since we already work with sets of connections, we do not have to pay an additional penalty for memory management when we switch to multi-criteria.
In a sense, we already have a multi-criteria optimization problem, we optimize for travel time but relax it with the train and transfer information.
However, for traffic days, the dominant ranges may work not well.
In this case we can e.g. split the connections by day or introduce a dominant linked list of connections, that are relevant from a given earliest departing transfer connection.

}
\end{document}